\documentclass[10pt,twoside, twocolumn]{IEEEtran}
\usepackage{graphicx}
\usepackage{amsmath}
\usepackage{amssymb}
\usepackage{bm}
\usepackage{color}
\usepackage{url}

\graphicspath{{./fig/}}

\newcommand{\qg}{{\bf g}}
\newcommand{\qh}{{\bf h}}

\newcommand{\qn}{{\bf n}}

\newcommand{\qr}{{\bf r}}

\newcommand{\qt}{{\bf t}}

\newcommand{\qv}{{\bf v}}
\newcommand{\qw}{{\bf w}}
\newcommand{\qx}{{\bf x}}

\newcommand{\qA}{{\bf A}}

\newcommand{\qI}{{\bf I}}

\newcommand{\qX}{{\bf X}}

\newcommand{\qzero}{{\bf 0}}

\newcommand{\be}{\begin{equation}} \newcommand{\ee}{\end{equation}}
\newcommand{\bea}{\begin{eqnarray}} \newcommand{\eea}{\end{eqnarray}}

\newtheorem{lemma}{Lemma}
\newtheorem{proposition}{Proposition}

\def\bl#1{\color{black}{#1}}
\def\rl#1{\color{black}{#1}}

\DeclareMathOperator{\g}{\mathbf{g}}
\DeclareMathOperator{\h}{\mathbf{h}}

\DeclareMathOperator{\bv}{\mathbf{v}}
\DeclareMathOperator{\w}{\mathbf{w}}

\DeclareMathOperator{\D}{\mathbf{D}}

\DeclareMathOperator{\0}{\mathbf{0}}

\DeclareMathOperator{\re}{Re}

\DeclareMathOperator{\diag}{diag}

\DeclareMathOperator{\tran}{{\scriptstyle T}}

\DeclareMathOperator{\st}{{s.t.}}

\begin{document}

\title{Information and Energy Cooperation in  Cognitive Radio Networks}
\markboth{\textit{A Manuscript Accepted in The IEEE Transactions on Signal Processing } }{} 
 \author{   Gan Zheng, {\it Senior Member, IEEE},  Zuleita   Ho, {\it Member, IEEE}, Eduard A. Jorswieck, {\it Senior Member, IEEE},
 and Bj$\ddot{\rm o}$rn Ottersten, {\it Fellow, IEEE}
 \thanks{
 Gan Zheng is with School of Computer Science and Electronic Engineering, University of Essex, UK, E-mail: {\sf ganzheng@essex.ac.uk.}
 He is also affiliated with Interdisciplinary Centre for Security, Reliability and Trust (SnT),
  University of Luxembourg, Luxembourg.}
  \thanks{Zuleita   Ho was with Communications Laboratory, Dresden University of Technology, Dresden,
  Germany. She is now a senior engineer in the Advanced Communications Laboratory, DMC R\&D Center, Samsung Electronics, Korea, E-mail:{\sf
  zuleita.ho@samsung.com}.}
\thanks{Eduard A. Jorswieck is with Communications Laboratory, Dresden University of Technology, Dresden, Germany, E-mail:{\sf  eduard.jorswieck@tu-dresden.de}. This work has been
performed in the framework of the European research projects DIWINE
and ACROPOLIS, which are partly funded by the European Union under
its FP7 ICT Objective 1.1 - The Network of the Future.}
 \thanks{Bj$\ddot{\rm o}$rn Ottersten is with the
Interdisciplinary Centre for Security, Reliability and Trust (SnT),
  University of Luxembourg, 
Luxembourg, E-mail: {\sf  bjorn.ottersten@uni.lu.}}
 \thanks{Copyright (c) 2013 IEEE. Personal use of this material is permitted. However, permission to use this material for any other purposes
 must be obtained from the IEEE by sending a request to pubs-permissions@ieee.org.}
 }

 \maketitle
\date{\today}
\maketitle
\begin{abstract}
 Cooperation between the primary and secondary systems can   improve the spectrum efficiency
in cognitive radio networks. The key
 idea is that the secondary system helps to boost the primary system's performance by relaying
and in return the primary system provides more opportunities
 for the secondary system to access the spectrum. In contrast to
most of existing works that only consider information cooperation,
 this paper studies  joint information and energy cooperation between the two systems,
i.e., the primary transmitter sends information for relaying and
 feeds the secondary system with energy as well. This is particularly useful when the secondary
transmitter has good channel quality to the primary receiver but
  is energy constrained. We propose and study three schemes that enable this cooperation.
Firstly, we assume there exists an ideal backhaul  between the two
  systems for information and energy transfer. We then consider {  two} wireless information and
energy transfer {  schemes} from the primary transmitter to the
secondary
  transmitter using power splitting and time splitting energy harvesting techniques, {  respectively}.
For each scheme, the optimal and  zero-forcing solutions are derived.
  Simulation results demonstrate promising performance gain for both systems due to the
additional energy cooperation. {   It is also revealed that the
power splitting scheme can achieve larger rate region than the time
splitting scheme when the efficiency of the energy transfer is
sufficiently large.}
 \end{abstract}


 \section{Introduction}
 \subsection{Motivation}
  Cooperative cognitive radio networks (CCRN)    have been a new paradigm to improve the spectrum efficiency of a cognitive radio (CR) system
 where  the primary and secondary systems actively seek opportunities to cooperate with each other.
 CCRN  have many advantages over existing non-cooperative CR schemes.
 It is a win-win strategy for both systems in the sense that the
secondary transmitter (ST) helps relay the traffic from the primary transmitter (PT) to the primary user (PU), and in return can utilize the primary
spectrum to serve its own secondary user (SU). This is especially preferred by the primary system when the PU's quality-of-service (QoS) cannot be
met by the primary system itself. Compared to the conventional interweave CR technique \cite{Mitola-2000} which is an opportunistic access scheme,
the cooperation scheme does not require the ST to wait and sense the spectrum holes for transmission; unlike the underlay technique
\cite{Haykin-2005} which sets limit on the interference to the primary system, the cooperation scheme focuses on the end performance, e.g., the PU
rate or the signal to interference plus noise ratio (SINR),  thus the ST is no longer restricted to transmit with low power.

 However, most existing CCRN  assume that the cooperation is only at the information level. One problem is that even when the ST has good channel
 quality to help serve the PU but is energy constrained, the cooperation is still not possible. This is a commonly seen situation when the ST is a
 low-power relay node rather than  a powerful base station (BS). This motivates us to propose the cooperation between the primary and
 secondary systems at both information and energy levels, i.e., the PT will transmit both information and energy to the ST, in exchange for the ST to relay
 the primary information. Compared to the existing CCRN with only information cooperation, this scheme creates even stronger incentives for
 both systems to cooperate and  substantially  improves the system overall spectrum efficiency. It can be seen as an enhanced win-win strategy.
 The energy cooperation can be enabled by the recently proposed energy harvesting or wireless energy transfer techniques
\cite{Zhang_11}. In the sequel, we will
 briefly review the literature about CCRN and energy cooperation.

 \subsection{Related Works}
 \subsubsection{CCRN}
 Early works about CCRN are mainly  from the viewpoint of  information theory \cite{Ephremides-07}-\cite{Viswanath-bound}
 assuming non-causal  primary message available at the ST, where   the ST employs dirty paper coding (DPC) to remove interference from the SU due to the
primary signal.   Using multiple antennas and non-causal primary message at the ST, the optimal beamforming is studied for both cases using DPC
\cite{Jorswieck-WCNC-12} and linear precoding \cite{Jorswieck-WSA-11}. However, these schemes require  non-causal primary information at the ST;
therefore they are hard to implement in practice and only provide  an outer bound on the achievable primary-secondary rate region.

As to more practical CCRN,   three-phase cooperation protocols
between primary and cognitive systems are proposed to exploit
primary resources in time and frequency domain
\cite{Spectrum-leasing-Simeone}, \cite{Spectrum-leasing-Su}. The ST
uses the first two phases to listen and forward the primary traffic;
in return, the last phase is exclusively reserved for the ST to
transmit its own signal to the SU. The use of multiple antennas and
beamforming at the ST provides additional degree of freedom for the
concurrent primary-cognitive cooperation. The zero-forcing (ZF)
beamforming technique and the optimal beamforming solution have been
studied in \cite{Manna-2011}\cite{CCRN-Song} and
\cite{Zheng_CCRN_13}, respectively. Different from the
single-antenna case, the ST with multiple antennas requires only two
phases: Phase I is the same as that in the single-antenna case while
in Phase II, the ST can both relay the primary signal and transmit
its own signal due to its ability of signal separation  in the
spatial domain. Recently full-duplex radio has been investigated in
\cite{Zheng_CCRN_FD_13} for CCRN which requires only one phase and
it can efficiently enlarge the achievable rate region. { Both the
uncoordinated underlay cognitive radio scenario and the coordinated
overlay cognitive radio scenario that consists of a message-learning
phase followed by a communication phase are studied in
\cite{Jorswieck-JWCN-13}.}

\subsubsection{Energy Cooperation}
Energy cooperation is a promising solution to prolong the network lifetime despite the possible loss during the process of energy transfer. In case
of power line systems, joint communication and energy cooperation is investigated in \cite{Zhang_Globecom} for the coordinated multi-point  downlink
cellular networks. The base stations (BSs) powered  by renewable energy  are connected by a power line to enable simultaneous data and energy
sharing. The proposed joint communication and energy cooperation solution are shown to substantially improve the downlink throughput for energy
harvesting (EH) systems, as compared to the case without    energy cooperation. In a similar scenario, the optimal energy cooperation algorithms
are designed in \cite{Zhang_WCNC} for both cases where the renewable energy profile and energy demand profile are deterministic and stochastic.

As to wireless energy transfer,  recently  the radio frequency (RF)
EH technology has emerged as a new solution where the
electromagnetic radiation in the environment is captured by the
receiver antennas and converted into useful energy. {  Thanks to
recent advances in antenna and rectenna circuit design, there has
been great progress towards improving the efficiency of wireless
energy transfer, for instance,    Powerharvester receivers provided
by   Powercast   can achieve conversion efficiency   as high as 70\%
in some scenarios \cite{powercase}. A sensor node powered by a
cellular Base Transceiver Station (BTS) at a distance of 200m from
the BTS was implemented in \cite{bts}}. RF-EH technique also enables
simultaneous transfer of information and energy using RF signals
\cite{Varshney-08}\cite{Grover-10}. Two practical receiver
structures to decode information and EH called ``time switching''
and ``power splitting'', are proposed in \cite{Zhang_11}. ``Power
splitting'' divides the received signal into two parts, one for
harvesting energy and the other for information decoding. ``Time
switching''  uses   dedicated time slots for harvesting energy and
the rest for data transmission.  Dynamic switching between
information decoding and RF EH is proposed in \cite{Zhang-12} then
further studied in \cite{Krikidis-12} for a cooperative relaying
scenario with a discrete-level battery at the RF-EH relay node.

When the wireless terminals have RF EH capabilities, energy cooperation provides more performance gain in additional  to the usual information
cooperation in cooperative communications.  Energy cooperation is considered in   \cite{energy_coop} for several basic multi-user network structures
 including relay channel, two-way channel and multiple access channel. The optimal energy management policies that maximize the system throughput
within a given duration  is studied.  A more relevant one to our
work is \cite{Zhou-13}, where an energy constrained relay node
harvests energy from the received RF signal and uses that harvested
energy to forward the source information to the destination,
therefore the relay does not need external power supply. Both time
switching and power switching relaying protocols  are proposed to
enable EH and information processing at the relay.  Outage capacity
and ergodic capacity are also derived.


 \subsection{Contribution}

 In this paper, we propose   two-level cooperation between the primary and the secondary systems  to achieve better use of the spectrum. The first one is information cooperation, where the PT broadcasts the primary signal and after receiving it, the
 ST retransmits it to the PU; the second one is energy cooperation, where the PT transmits power to the ST via either cable or wireless medium,
 such that the ST can obtain extra power to help the PT, as well as serve its own SU. The ST may have good channel condition to the PU, but lacks spectrum and energy, therefore,
 the two-level cooperation substantially increases the chance that the ST can assist the primary transmission and use the primary spectrum.
 We assume  the ST is  equipped with multiple antennas, and deals with the beamforming design to characterize the achievable primary-secondary rate region.
 In particular, we study the problem of maximizing the SU rate subject to the PU rate and ST power (including harvested power) constraints by    optimizing the beamforming design  at the ST for EH and relay processing.

 We propose three schemes that enable  information as well as energy  cooperation: i) ideal cooperation where we assume the primary information is
 non-causally known at the ST 
 and the transmit power can be   shared between the PT and the ST;  ii) power splitting scheme where the ST uses part of received signal for information
 decoding and the rest for energy harvesting; and iii) time
 splitting scheme where   a fraction of time is reserved for wireless energy transfer from the PT  to the ST and the rest of time is used for
 information listening and forwarding. For each scheme, we propose efficient algorithms to optimally solve the above mentioned optimization problem. In
 addition, we derive low-complexity solutions based on ZF criterion,  {  which provide some insights on the impacts of system parameters}.

 \subsection{Notations}
Throughout this paper, the following notations will be adopted. Vectors and matrices are represented by boldface lowercase and uppercase letters,
respectively.    $\|\cdot\|$ denotes the Frobenius norm. $(\cdot)^\dag$ denotes the Hermitian operation of a vector or matrix.   $\qA\succeq \qzero$
means that $\qA$ is positive semi-definite. $\qI$ denotes an identity matrix of appropriate dimension. ${\tt E}[\cdot]$ denotes the expectation.
  ${\bf x}\sim\mathcal{CN}({\bf m},{\bf\Theta})$ denotes a vector $\qx$ of complex Gaussian elements with a mean vector of ${\bf m}$ and a
covariance matrix of ${\bf\Theta}$.  $\Pi_{\qX}$ denotes the orthogonal projection onto the column space of  $\qX$ while $\Pi_{\qX}^\bot$ denotes the
orthogonal projection onto the orthogonal complement of the column space of $\qX$. We further define $[x]_0^1 \triangleq \min(1,\max(0,x))$.
$\diag(\bv)$ denotes a diagonal matrix with diagonal elements as the elements of $\bv$.

\section{General System Setting}
\begin{figure}[h]
  \vspace{-5mm}
  \centering
  \includegraphics[width=3.5in]{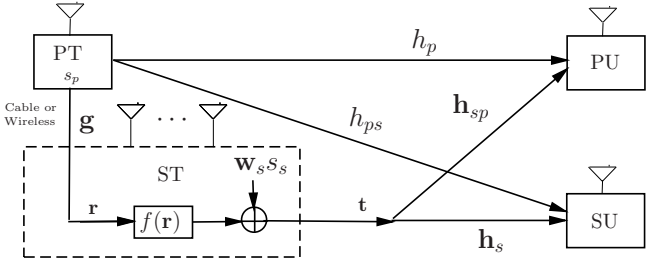}
  \caption{Energy and information cooperation in cognitive radio.}\label{fig:sys:model}
\end{figure}

 We consider cooperation between a primary system and a secondary system in cognitive radio networks, as depicted in Fig. \ref{fig:sys:model}.
 The  primary system consists of a primary transmitter
 (PT) and a primary user (PU), while the secondary system has a secondary transmitter (ST) who serves a
 secondary user (SU). All terminals have a single antenna except that the ST has $N$ antennas.
 The PT intends to send $s_p$ to the PU while the ST transmits signal $s_s$ to the SU with appropriate power. We consider a scenario that
 the primary link is in outage status {\bl when its rate demand cannot be met via the direct link,}
 thus it becomes necessary for the PT to cooperate with  the ST in order to meet the PU's QoS
 requirement. {  Without loss of generality,  we assume the communication duration  $T$ is normalized to be unity. }

 Some common system parameters are introduced as follows.
{\bl{$h_p$, $\qh_s$, $\qh_{sp}$ and $h_{ps}$ are used to denote the
PT-PU, ST-SU,  ST-PU and PT-SU channels},
 respectively.     The PT is connected to the ST either via cable or wireless channel $\qg$.
  {  The PT's total energy (or average transmit power) is $P_p$ and the PU's rate requirement is $r_p$ bps/Hz.
 The ST itself has an initial   total energy of $P_{s0}$} and  further receives/harvests energy from the
 PT.}
  All channels and noise elements are assumed to be statistically independent of each other.
{ We assume that global perfect channel state information (CSI) is available at the ST.}
 After the ST receives both information and energy from the PT, it processes the primary signal, harvests energy then uses the harvested
 energy together with its own {\bl energy}, to serve the SU and relay the signal to the PU. Amplify-and-forward (AF) relaying protocol is employed by the ST.

 We explicitly consider two components in the noise received at a terminal: one is the received thermal noise and the other is due to RF to baseband conversion,
 both are modeled as   zero-mean additive white Gaussian noise (AWGN) with variances of $N_0$ and $N_C$
respectively. Assuming that they are independent,
 {\bl we may consider both types of noise if
 possible and define the combined received noise power as $\tilde N_0=N_0 +
 N_C$.}


 There are different approaches that facilitate the information and energy transfer from the PT to the ST. In the sequel, we will introduce three specific
 cooperation schemes and find their optimal as well as low-complexity
 solutions. {  For fairness, the same amount of energy (for both the PT and the ST) is used  in all schemes.}

\section{Ideal Primary-Cognitive Cooperation}
\subsection{System Model and Problem Formulation}
 We first look into the ideal cooperation between the PT and the ST for information and energy transfer, where  the ST has non-causal information
  about the primary signal  and obtains energy from the PT via reliable backhaulling,  for instance, cable. Note that although in practice, this cooperation scheme
 is either too difficult or too costly to implement,   it provides a performance upper bound for practical cooperation protocols.

  Since the primary message $s_p$ (${\tt E}[|s_p|^2]=1$) is non-causally known at the ST, it can employ the DPC technique to encode the primary signal and
  superimpose its own  secondary signal $s_s$ (${\tt E}[|s_s|^2]=1$) such
  that no primary interference is introduced at the SU
  \footnote{
Note that the   ST can also pre-cancel the interference at the PU
caused by its own secondary signal, however,   this could lead to
performance degradation because the SU will receive interference
from both the PT and  the ST.}.
  The received signal at the PU is
 \be y_p=   (\sqrt{(1-\beta) P_p}h_p +  \qh_{sp}^\dag \qw_p)s_p +  \qh_{sp}^\dag \qw_s s_s + n_{p},\ee
 where $\qw_p$ is the beamforming vector used by the ST to forward the primary signal, and $n_p\in \mathcal{CN}(0,\tilde N_0)$ is the combined received
 noise at the
 PU. $\beta P_p (0\le \beta\le 1)$ denotes the amount of {  energy} transferred to the ST  and received as $\eta\beta P_p$,
 where $\eta$ is the efficiency of energy transfer. The ST then has the total power of $P_{s0} + \eta\beta P_p$ to serve both the PU and the SU.

 Due to the use of DPC, the SU receives
 \be
    y_s =  \qh_{s}^\dag \qw_s s_s + n_s,
 \ee
 where $n_s\in \mathcal{CN}(0,\tilde N_0)$ is the combined received  noise at the SU.
  Then the primary and secondary received signal to interference-plus-noise ratios (SINRs), are, respectively,
 \be
    \Gamma_p=  \frac{|\sqrt{(1-\beta) P_p}h_p + \qh_{sp}^\dag \qw_p|^2}  {|\qh_{sp}^\dag \qw_s |^2 + \tilde N_0}, ~~\mbox{and}~~ \Gamma_s= \frac{|\qh_{s}^\dag \qw_s|^2}{
    \tilde N_0}.
 \ee

 It is easy to see that the optimal $\qw_p$ admits the form $\qw_p = \sqrt{q_p} \frac{\qh_{sp}}{\|\qh_{sp}\|} e^{j\theta}$, where $\theta$ is chosen for coherent reception and
 $q_p\triangleq \|\qw_p\|^2$. As a result, the achievable PU rate is
 \be
    R_p = \log_2 (1+\gamma_p) = \log_2\left( 1+ \frac{(\sqrt{(1-\beta) P_p}h_p + \sqrt{q_p}\|\qh_{sp}\|) ^2}  {|\qh_{sp}^\dag \qw_s |^2 +
    \tilde N_0}\right).
 \ee

  The   problem of maximizing the SU rate subject to PU rate and ST power constraints is written as
 \begin{subequations}\label{eqn:prob:DPC}
  \begin{align}
      \max_{\qw_s,q_p, \beta} \hspace{1cm} &  |\qh_{s}^\dag \qw_s| \\
    \mbox{s.t.} \hspace{1cm} &
\frac{(\sqrt{(1-\beta) P_p}|h_p| + \sqrt{q_p} \|\qh_{sp}\|)^2}  {|\qh_{sp}^\dag \qw_s |^2 + \tilde N_0}
\ge 2^{r_p}-1,\label{eqt:con1}\\
    & \|\qw_s\|^2 + q_p\le  P_{s0} + \beta \eta P_p, \label{eqt:con2}\\
    & 0\le \beta\le 1, q_p>0. \label{eqt:con3}
  \end{align}
 \end{subequations}


\subsection{Feasibility}
 Before solving (\ref{eqn:prob:DPC}), in Proposition 1 we first give the condition under which it is
 feasible and {\bl{the proof is given in   Appendix \ref{app:feas}.}}
\begin{proposition}\label{prop:feas}
 Problem  (\ref{eqn:prob:DPC}) is feasible if and only if the PU rate requirement $r_p$ is {\bl not larger} than
 \be\label{eqn:Rpmax:DPC}
    R_{p,\max}\triangleq\left\{\begin{array}{cc}
              \log_2\left(1+  \frac{(\sqrt{P_p}|h_p| + \sqrt{P_{s0}}\|\qh_{sp}\| )^2}  {
    \tilde N_0}\right), & \\ \mbox{if~}P_p \eta^2 \|\qh_{sp}\|^2 < P_{s0} |h_p|^2; & \\
 \log_2\left(1+  {\frac{  P_p \eta  + P_{s0}  }{
\eta^2 \|\qh_{sp}\|^2 + \eta   |h_p|^2}}\frac{\left(|h_p|^2 +
{\eta} \|\qh_{sp}\|^2 \right)^2}{\tilde N_0}\right), &
\\ \mbox{otherwise}.&
            \end{array}\right.
 \ee
\end{proposition}
  According to Proposition 1, when  $P_p \eta^2
\|\qh_{sp}\|^2 < P_{s0} |h_p|^2$ which means that the transferred
power from the PT to the ST cannot bring sufficient performance
gain,   no energy transfer is needed.  This may happen when primary
power is limited, the secondary power is abundant, or the efficiency
of the power transfer is too low, etc.

\subsection{Simplified {\bl Characterization of \eqref{eqn:prob:DPC}} }

{\bl{ Problem \eqref{eqn:prob:DPC} is equivalent to the following
convex problem:
\begin{subequations}\label{eqt:cov_opt}
 \begin{align}
  \max_{\w_s , \bv} \hspace{1cm} & \re\left(\h_s^{\dagger} \w_s\right) \label{eqt:obj}\\
  \st \hspace{1cm} &  \left( \g^{\tran} \bv  \right)^2
\geq (2^{r_p}-1) \left( |\h_{sp}^{\dagger} \w_s|^2 + \tilde{N}_0 \right), \label{eqt:soc}\\
 & \| \w_s\|^2 + \bv^{\tran} \D \bv \leq P_{s0} +  \eta P_p,\\
 & \bv \geq \0, [\bv]_{1}\leq 1.
 \end{align}
\end{subequations} where
$\g^{\tran}=[\sqrt{P_p} |h_p|, \|\h_{sp}\|]$,
  $\D= \diag(\eta P_p, 1), {  \qv\triangleq [\sqrt{1-\beta},\sqrt{q_p}]}$. Optimization problem in \eqref{eqt:cov_opt} is a
second order cone problem (SOCP) and is convex. It can be solved
very efficiently. The steps of recasting \eqref{eqn:prob:DPC} to the
convex problem \eqref{eqt:cov_opt} are given in Appendix
\ref{app:convex}.}}

 Although problem \eqref{eqn:prob:DPC} can be manipulated as a convex problem, it does not offer much
insight into the structure of the solution. In the following, from
the characterization of the optimal solution structure, we identify
very efficient solutions. To this end, we first show that at the
optimality point, the two constraints in \eqref{eqn:prob:DPC} are
active. If the power constraint is not active at the optimality
point, one can update the beamforming vector $\w_s$ to $\w_s'= \w_s
+ \tau \Pi_{\h_{sp}}^{\perp} \h_s$ with $\tau$ being a very small
scalar, which increases the objective value while keeping the first
constraint unchanged. This contradicts  the optimality point
assumption. Then if the first constraint is not active, one can
decrease the value of $q_p$ and increase the value of $\beta$ such
that the first constraint is active but the power constraint is not
active. This leads us back to the previous case and a contradiction
of the optimality assumption results. Hence at the optimality point,
the two constraints are always active.
The main result of the simplified optimization problem is given in
Proposition 2 {\bl and its proof is provided in Appendix
\ref{app:feas2}.}

\begin{proposition}\label{prop:2d}
 Problem (\ref{eqn:prob:DPC}) is equivalent to the
 problem (\ref{eqn:prob:DPC:2D})  at the top of next page.
 \begin{figure*}
  \bea\label{eqn:prob:DPC:2D}
    \max_{\beta, q_p} &&      \sqrt{ \frac{|\sqrt{(1-\beta) P_p}\frac{|h_p|}{\|\qh_{sp}\|} + \sqrt{q_p}|^2}  {2^{r_p}-1 }  - \frac{\tilde N_0}{\|\qh_{sp}\|^2}}
     \|\Pi_{\qh_{sp}} \qh_s\| \\
     &&  + \sqrt{(P_{s0} + \beta \eta P_p) - q_p - \frac{|\sqrt{(1-\beta) P_p}\frac{|h_p|}{\|\qh_{sp}\|} + \sqrt{q_p}|^2}  {2^{r_p}-1 }  - \frac{\tilde N_0}{\|\qh_{sp}\|^2}}
\|\Pi_{\qh_{sp}}^\bot \qh_s\| \notag \\
\mbox{s.t.}&& 0\le \beta\le 1, q_p>0. \notag
 \eea
 \end{figure*}
 \end{proposition}
  {  Problem (\ref{eqn:prob:DPC:2D}) is not convex in general, so we
  propose to find its optimal solution via 2-D search.}

We illustrate the feasibility region in Figure
\ref{fig:feasibility_reg2}, where {\rl{ randomly chosen channel
realizations and system settings are given by $\qh_s =[-0.0823 +
1.3427i, -0.6438 - 0.4291i, 0.4338 - 0.2197i]^T$, $\qh_{sp} =[0.5345
- 0.8716i, 0.2872 - 0.4043i, 0.0951 - 0.3264i]^T$, $h_p = -0.4692 +
0.8665i$, $\gamma_0=5$, $P_p= P_{sp} = 10$, $\eta = 0.8$ and $N= 3$.
}} The feasibility region of two target SINR values are shown. {
Each blue ring circle is composed of tuples of $(|\h_{sp}^\dag
\w_s|,|\h_{s}^\dag \w_s|^2)$ and  is obtained by varying $\w_s$ and
keeping $q_p$ and $\beta$ fixed.}  When $q_p$ increases and $\beta$
fixed, constraint \eqref{eqt:con1} on $\h_{sp}^{\dagger}\w_s$ is
more relaxed. This is illustrated by the red lines, each line
corresponding to a given pair $q_p, \beta$. When $q_p$ increases,
the red line moves up. Similarly, when  $q_p$ increases and $\beta$
fixed, \eqref{eqt:con2} becomes more strict and thus the blue region
shrinks. Similar behavior can be observed if $\beta$ varies and
$q_p$ is kept constant. The feasibility region is the blue region
under the corresponding red line. Hence, there is a conflict between
constraints \eqref{eqt:con1} and \eqref{eqt:con2}. In Figure
\ref{fig:feasibility_reg2}, the black cross marker shows the optimal
point for each given tuple of $(q_p, \beta)$. The optimal point is
the {  black cross} with the largest x-coordinate. {  In Figure
\ref{fig:ch_pow}, we collect the optimal points (shown in black
crosses in both Figure 2 and 3) for each paired value of $(q_p,
\beta)$. Out of these optimal points of each realization of
$(q_p,\beta)$, we mark the optimal point of the whole set as a red
square. If the zero-forcing scheme is implemented, the achievable
points by definition are always on the x-axis and the optimal point
is marked as a blue triangle in Figure \ref{fig:ch_pow}.}

\begin{figure}
 \begin{center}
  \includegraphics[width=3.1in]{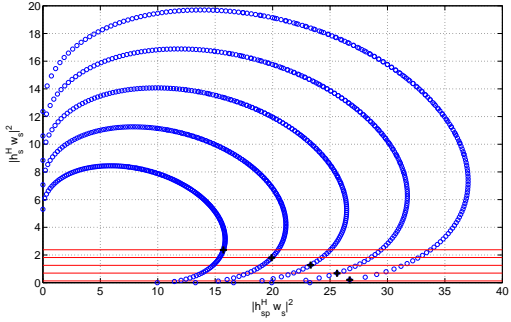}
\caption{  The channel power region of $| \qh_s^\dagger \w_s |^2$
against $| \qh_{sp}^\dagger \w_s |^2$. ${(\cdot)}^H$ represents
Hermitian operation in the figure and is denoted as ${(\cdot)}^\dag$
in the main text. The red lines show the constraint moving up when
$q_p$ increases and $\beta$ fixed. At the same time, the channel
power region, shown in blue,  shrinks. The feasibility region is the
blue region under the corresponding red line. Within the feasibility
regions, the point with the maximum x-coordinate is marked with a
black cross. \label{fig:feasibility_reg2}}
 \end{center}
\end{figure}

\begin{figure}
 \begin{center}
  \includegraphics[width=3.2in, keepaspectratio]{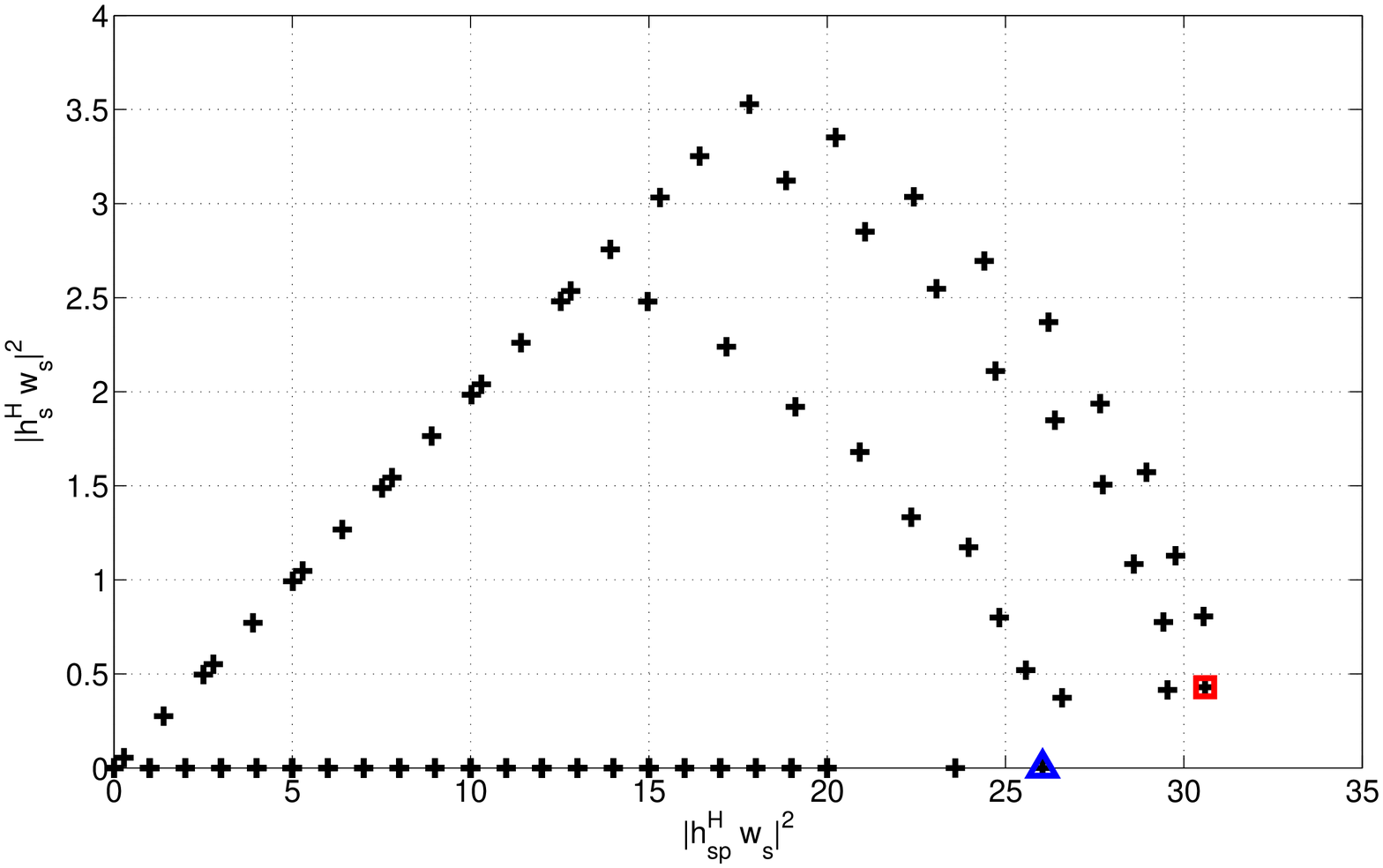}
\caption{   The optimal point for each pair of values of $(q_p,
\beta)$ is collected and shown in black crosses. The optimal point
that attains the largest value of $|h_{sp}^\dag \w_s|^2$, is marked
as a red square. When the zero-forcing scheme in Section III-D is
used, the optimal point is marked as a blue triangle.
 \label{fig:ch_pow}}
 \end{center}
\end{figure}

 \subsection{ZF Solution}\label{sec:zf}
  Here we study a  suboptimal yet closed-form  solution with ZF constraint on the interference power from the ST to the PU, i.e., $\qh_{sp}^\dag\qw_s=0$. To satisfy this, the
  beamforming vector $\qw_s$ is chosen as
   \be\label{eqn:DPC:ZF:ws}
     \qw_{s,ZF} = \sqrt{q_s} \frac{ \left(\qI - \frac{\qh_{sp}\qh_{sp}^\dag}{\|\qh_{sp}\|^2}\right) \qh_s}{\|\left(\qI - \frac{\qh_{sp}\qh_{sp}^\dag}{\|\qh_{sp}\|^2}\right)
    \qh_s\|},
 \ee
 and the resulting SU channel gain is
  \be
    |\qw_{s,ZF}^\dag\qh_{s}|^2 = q_s\|\qh_{s}\|^2(1-\delta^2),  \delta^2 \triangleq \frac{|\qh_{sp}^\dag\qh_{s}|^2}{\|\qh_{sp}\|^2\|\qh_s\|^2}.
 \ee
 As a result, the optimization problem is formulated as
{\small \bea
    \max_{q_s,q_p, \beta} &&  q_s\\
    \mbox{s.t.} && \frac{(\sqrt{(1-\beta) P_p}|h_p| + \sqrt{q_p}\|\qh_{sp}\|)^2}  { \tilde N_0}\ge 2^{r_p}-1,\label{eqn:dpc:zf:c1}\\
    && q_s + q_p\le  P_{s0} + \beta \eta P_p,\label{eqn:dpc:zf:c2}\\
    && 0\le \beta\le 1, q_s\ge 0, q_p\ge 0.\notag
 \eea}
 Due to the fact that both constraints (\ref{eqn:dpc:zf:c1}) and (\ref{eqn:dpc:zf:c2}) should hold with equality, $q_s$ can be expressed as
 {\small \bea
    & &q_s=  P_{s0} + \beta \eta P_p  - q_p\\
     &&= P_{s0} + \beta  P_p \left(\eta +  \frac{|h_p|^2}{\|\qh_{sp}\|^2}\right)  - \\
     &&\frac{(2^{r_p}-1)\tilde N_0  +  P_p|h_p|^2 -2 \sqrt{(2^{r_p}-1)\tilde N_0}\sqrt{(1-\beta) P_p}|h_p| }{
    \|\qh_{sp}\|^2}.\notag\label{eqn:beta:qs}
 \eea}

 By setting the first-order derivative of (\ref{eqn:beta:qs}) to be zero,
 we can derive the optimal $\beta$ as
   \be\label{eqn:dpc:zf:beta}
       \beta^*  = \left[1- \frac{    {(2^{r_p}-1)\tilde N_0}   }{  {P_p}\left(|h_p|+ \frac{\eta \|\qh_{sp}\|^2}{|h_p|  } \right)^2 }\right]_0^1.
  \ee

  {  The expression (\ref{eqn:dpc:zf:beta}) verifies the intuition that  the optimal $\beta$ is
   an increasing function of $P_p$, $\|\qh_{sp}\|^2$   and $\eta$. If  the PT has abundant power, then it is more likely to transfer energy to the ST.
   On the other hand, if the ST-PU link is weak  or the efficiency of energy transfer is low, it is not worth transferring too much energy to the
  ST. There is an interesting observation about $|h_p|$.  If $|h_p|$ is close to zero, $\beta$ approaches 1, which means that the primary system relies on the ST to forward its signal, therefore transfers all its energy to the ST.
  As the primary channel becomes better or  $|h_p|$ increases but is below the threshold $\sqrt{\eta} \|\qh_{sp}\|$, $\beta$ is a  decreasing function of $|h_p|$;
  once  $|h_p|$ exceeds the threshold $\sqrt{\eta} \|\qh_{sp}\|$, $\beta$ becomes an increasing function of $|h_p|$ and this is because the primary channel is good enough therefore the primary
  system can help the secondary transmission. }

{    The channel power values achieved by the ZF solution achieves
are shown in Figure \ref{fig:ch_pow}. The corresponding optimal
solution is marked blue. } The ZF solution is simple and we observe
that its performance is quite good {\rl{in this example. }}

    {\bl The implementation of the ideal cooperation requires cable and common energy source, as well as signal processing and coding capabilities for DPC}. In the following two sections, we consider two practical energy and information
   cooperation schemes. We assume that the ST adopts the AF relaying protocol to forward the primary signals instead of processing non-causal primary
   information.

\section{Power Splitting Cooperation -- System Model and Optimization}

\subsection{System Model and Problem Setting}
 In this section, we assume the ST   first listens to the primary transmission via the channel $\qg$ then forwards it to the PU, therefore
   two channel phases are required to complete the communications.
 {  In Phase I, the PT broadcasts its data $s_p$ with   power $2P_p$ where the factor 2 is because the PT only transmits during the first half duration,
 then the received signals at the PU and the ST are, respectively,
 \be y_{p1}=  \sqrt{2P_p} h_p s_p + n_{p1},~~ \mbox{and}~~\qr= \sqrt{2P_p}  \qg  s_p +
 \qn_R,\ee}
 where  $ n_{p1}\in \mathcal{CN}(0,\tilde N_0)$ is the combined noise at the PU while
  $\qn_R \in \mathcal{CN}(\qzero,N_0 \qI)$ is the thermal noise received at  ST, respectively.

 \begin{figure}[]
        \centering
                 \includegraphics[width=3.6in]{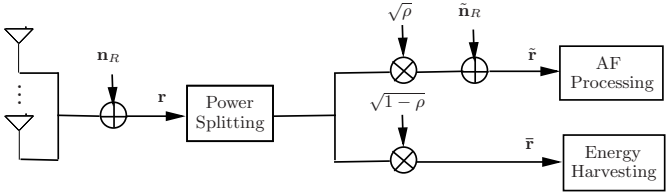}
        \caption{The power splitting EH technique at the ST. }
        \label{fig:PS}
\end{figure}
To forward primary information as well as  harvest RF energy at the ST, the practical power
splitting technique \cite{Zhang_13} is used, which
is depicted in Fig.  \ref{fig:PS} and works as follows. The ST splits the RF signal into two portions: one for forwarding to the PU after AF
processing  and the other for  harvesting energy, with relative power of $\rho$ and $1-\rho$, respectively. The signal for AF processing will be
converted from the RF   to the baseband, and this results in the received signal \be
    \tilde \qr = \sqrt{\rho} \qr + \tilde \qn_R,
\ee where $\tilde \qn_R\sim \mathcal{CN}(\qzero,N_C \qI)$ is the complex AWGN during the RF to baseband conversion.
 The ST  processes the received signal and produces $f(\qr)= \qA\tilde\qr$. Without loss of optimality, it has been shown that the optimal $\qA$ has the structure of
  $\qA = \qw_p \qg^\dag$  according to \cite{Zheng_CCRN_13}, where $\qw_p$ is a new transmit beamforming vector to be optimized. This is also
  intuitive because there is a single primary data stream, the best reception strategy for
   the ST is to use maximal ratio combining (MRC).

The signal for EH is simply \be
    \bar \qr = \sqrt{1-\rho} \qr = \sqrt{1-\rho} \left(  \qg \sqrt{2P_p}  s_p + \qn_R\right).
\ee Assuming the {  energy} transfer efficiency of $\eta$, the
amount of the harvested energy is
 \be
    P_{EH} = \frac{\eta(1-\rho)(2P_p\|\qg\|^2+ N_0)}{2}.
\ee {  Therefore, the ST will have a total transmit power of
 $2P_{s0}+\eta(1-\rho)(2P_p\|\qg\|^2+ N_0)$ where   the factor 2 is due to the fact that the ST only
transmits signals in the second half of the communication time.}

In Phase II, the ST superimposes the relaying signal $f(\qr)$ with
its own data $s_s$   using the cognitive beamforming vector $\qw_s$,
then transmits it to both the PU and the SU. {\bl Note that DPC is
not used at the ST.} In this phase, the PT remains idle.

The ST's transmit signal is
 \bea
      \qt &=  &          \qw_s s_s+  \qw_p\qg^\dag \tilde\qr \\& =&
 \qw_s s_s+  \sqrt{2\rho P_p}\qw_p \|\qg\|^2  s_p +\sqrt{\rho}\qw_p\qg^\dag \qn_R + \qw_p\qg^\dag\tilde
 \qn_R,\notag
 \eea
 with average power
  \be
    p_R = {\tt E} \| \qt\|^2 = \|\qw_s\|^2 + (2P_p \rho \|\qg\|^4 + \rho\|\qg\|^2N_0  + \|\qg\|^2 N_C)\|\qw_p\|^2.
     \ee
 The received signal at the SU is
 \bea
   y_s &=&\qh_{s}^\dag\qt + n_s\notag\\
    &=&
    \qh_{s}^\dag\qw_s s_s +
     \sqrt{\rho}\qh_{s}^\dag\qw_p\|\qg\|^2 s_p + \sqrt{\rho}\qh_{s}^\dag\qw_p\qg^\dag\qn_R\notag\\
     && +  \qh_{s}^\dag\qw_p\qg^\dag\tilde \qn_R+ n_{s},
\eea where $ n_{s}\in \mathcal{CN}(0,\tilde N_0)$ is the combined noise at the SU.
 The received  SINR  at SU   is then expressed as
 \be
    \Gamma_s =\frac{|\qh_{s}^\dag\qw_s|^2}{ ( 2P_p \rho \|\qg\|^4 + \rho \|\qg\|^2 N_0  +  \|\qg\|^2 N_C) |\qh_s^\dag\qw_p|^2+  \tilde N_0 },
\ee and the achievable SU rate is $R_s =
\frac{1}{2}\log_2(1+\Gamma_s)$ where the factor $\frac{1}{2}$ arises
due to the two orthogonal channel uses. The received signal at the
PU is
 \bea
     y_{p2} &=&\qh_{sp}^\dag\qt+n_{p2} \notag\\
    &=&    \qh_{sp}^\dag  \qw_s s_s+ \rho\qh_{sp}^\dag \qw_p\|\qg\|^2 s_p +  \rho \qh_{sp}^\dag \qw_p\qg^\dag \qn_R  \notag\\ && + \qh_{sp}^\dag\qw_p\qg^\dag\tilde \qn_R+
    n_{p2},
\eea where $n_{p2}\in \mathcal{CN}(0,\tilde N_0)$ is the combined noise at the PU during Phase II.

 Applying the MRC strategy to $y_{p1}$ and $y_{p2}$, the received SINR of the PU is the sum  of two channel uses,
 and consequently, the achievable PU rate is
 \bea
    &&R_p=
    \frac{1}{2}\log_2\Big(1+\frac{2P_p |{h}_p|^2}{\tilde N_0}
    + \\&&
      \frac{  2P_p \rho\|\qg\|^4| \qh_{sp}^\dag\qw_p|^2}  { |\qh_{sp}^\dag\qw_s|^2+
    (\rho\|\qg\|^2N_0 + \|\qg\|^2 N_C)|\qh_{sp}^\dag \qw_p|^2+ \tilde N_0}\Big).\notag
 \eea

 Next we can formulate the problem of  maximizing the SU rate $R_s$ subject to the PU's
rate constraint $r_p$ and the ST's transmit
 power constraint $2(P_{s0}+P_{EH})$, by jointly  optimizing
 the power splitting parameter $\rho$, the cognitive beamforming vector $
\qw_s$, and the forwarding
 beamforming vector $\qw_p$.
 Using the monotonicity between the received SINR and the achievable rate, the optimization problem can be written as
  {\small    \bea\label{eqn:prob:rate:max:AF:v0}
      \max_{\qw_s,\qw_p,\rho} && \frac{|\qh_{s}^\dag\qw_s|^2}{ ( 2P_p \rho \|\qg\|^4 + \rho \|\qg\|^2N_0  +  \|\qg\|^2 N_C) |\qh_s^\dag\qw_p|^2+  \tilde N_0 } \\
    \mbox{s.t.}
   &&  \frac{  |\qh_{sp}^\dag\qw_p|^2}    {  |\qh_{sp}^\dag\qw_s|^2   + \tilde N_0}
    \ge \frac{\gamma_{p}^{'}}{ \rho\|\qg\|^4-\gamma_{p}^{'}(\rho\|\qg\|^2N_0 + \|\qg\|^2 N_C)}, \notag \\
&&  \|\qw_s\|^2 + (2P_p \rho \|\qg\|^4 + \rho\|\qg\|^2N_0  + \|\qg\|^2 N_C)\|\qw_p\|^2 \le \notag\\ &&   2P_{s0}+ \eta(1-\rho)(2P_p\|\qg\|^2+ N_0) , \notag\\
&& 0\le \rho\le 1,\notag
 \eea}
where we have defined  $\gamma_{p}^{'}\triangleq \frac{2^{2
r_p}-1}{2P_p}-\frac{|h_p|^2}{\tilde N_0}$.

\subsection{Feasibility Check}
 Before solving   problem (\ref{eqn:prob:rate:max:AF:v0}), we first investigate its feasibility,
and this can be achieved by finding the maximum PU rate $R_P$ or equivalently $\gamma_p'$.
 To achieve the maximum $\gamma_p'$, we set $\qw_s=\qzero$ then we have $\|\qw_p\|^2 =  \frac{2P_{s0}+ \eta(1-\rho)(2P_p\|\qg\|^2+ N_0)}{2P_p \rho \|\qg\|^4 + \rho\|\qg\|^2N_0  + \|\qg\|^2 N_C}$,
 and reach an optimization problem about $\rho$ below:
 {\small
   \bea\label{eqn:prob:rate:max:AF:v4}
      \max_{\rho} && \frac{ \rho\|\qg\|^4}{\frac{\tilde N_0(2P_p \rho \|\qg\|^4 + \rho\|\qg\|^2N_0  + \|\qg\|^2 N_C)}{   \|\qh_{sp}\|^2  \left(2P_{s0}+ \eta(1-\rho)(2P_p\|\qg\|^2+
      N_0)\right) } +  (\rho\|\qg\|^2N_0 + \|\qg\|^2    N_C)}\notag\\
    \mbox{s.t.}&& 0\le \rho\le 1.
 \eea}
The {\bl unique} optimal $\rho^*$ can be computed in closed form,
despite its complicated expression. For details, please see Appendix
\ref{app:feas_pow}. While $\rho^*$ corresponds to a maximum PU rate
$R_P^*$, we can choose any rate smaller than $R_P^*$ in solving
(\ref{eqn:prob:rate:max:AF:v0}).

\subsection{The Optimal Solution}
 Assuming problem (\ref{eqn:prob:rate:max:AF:v0}) is feasible, we study how to find its optimal solution.

 By change of variable $\qw_p: = \sqrt{2 P_p \rho \|\qg\|^4 + \rho \|\qg\|^2N_0  +  \|\qg\|^2 N_C} \qw_p$, we write (\ref{eqn:prob:rate:max:AF:v0})
 in an equivalent but more compact form as
   {\small \bea\label{eqn:prob:rate:max:AF:v1}
      \max_{\qw_s,\qw_p,\rho} && \frac{|\qh_{s}^\dag\qw_s|^2}{  |\qh_s^\dag\qw_p|^2+  \tilde N_0 } \\
    \mbox{s.t.}
   &&  \frac{  |\qh_{sp}^\dag\qw_p|^2}    {  |\qh_{sp}^\dag\qw_s|^2   + \tilde N_0}
    \ge \gamma_p^{''}, \notag \\
&&  \|\qw_s\|^2 + \|\qw_p\|^2 \le   2P_{s0}+ \eta(1-\rho)(2P_p\|\qg\|^2+ N_0),\notag\\
&&0\le \rho\le 1,\notag
 \eea}
where $\gamma_p^{''} \triangleq \frac{(2P_p \rho \|\qg\|^2 + \rho
N_0  + N_C)\gamma_{p}^{'}}{ \rho\|\qg\|^2-\gamma_{p}^{'}(\rho N_0 +
    N_C)}$. We find the following lemma  useful to solve (\ref{eqn:prob:rate:max:AF:v1}).

  \begin{lemma}\label{lemma:gam2}
     Consider {\bl a general} maximization problem below:
    \bea\label{eqn:prob:rate:max:basic}
    \max_{\qw_1,\qw_2} && \frac{|\qh_2^\dag\qw_2|^2}{\sigma^2+
    |\qh_2^\dag\qw_1|^2}\\
    \mbox{s.t.} && \frac{|\qh_1^\dag\qw_1|^2}{\sigma^2+
     |\qh_1^\dag\qw_2|^2}\ge \gamma_1\notag\\
     &&\|\qw_1\|^2+  \|\qw_2\|^2\le P_C,\notag
    \eea
    where   $\qh_1,\qh_2$ are $N\times 1$ vectors and $\gamma_1, P_C,\sigma^2$ are positive scalars.
    Define $\zeta^2 \triangleq \frac{|\qh_1^\dag\qh_2|^2}{\|\qh_1\|^2\|\qh_2\|^2}$. Suppose (\ref{eqn:prob:rate:max:basic}) is feasible and
      its optimal objective value  is  $\gamma_2$, then   $\gamma_2$   is uniquely determined by the following equation
    set:
\bea \left\{
\begin{array}{ll}
  \lambda_1& =\frac{ \gamma_1\sigma^2(\sigma^2+\lambda_2\|\qh_2\|^2)}{\|\qh_1\|^2(\sigma^2+\lambda_2\|\qh_2\|^2(1-\zeta^2))}\label{eqn:lambda:c},\\
      \lambda_2& =\frac{\gamma_2\sigma^2(\sigma^2+\lambda_1\|\qh_1\|^2)}{\|\qh_2\|^2(\sigma^2+\lambda_1\|\qh_1\|^2(1-\zeta^2))},\\
     \lambda_1+\lambda_2 &= P_C,\label{eqn:lambda:pow}
     \end{array}\right.
\eea
where $\lambda_1,\lambda_2$ are {  dual} variables.
\end{lemma}
\begin{proof}
See Appendix A in \cite{Zheng_CCRN_FD_13}.
\end{proof}

 {  Given $\rho$, using Lemma \ref{lemma:gam2}, the dual variables ($\lambda_1,\lambda_2$) of (\ref{eqn:prob:rate:max:AF:v1}) are identified.
 Then the  optimal $\qw_s, \qw_p$ can be expressed as
\bea \left\{
\begin{array}{ll}
  \qw_1& = \sqrt{p_1}\frac{ \left(\sigma^2\qI + \lambda_2\qh_2\qh_2^\dag\right)^{-1}\qh_1}{\|\left(\sigma^2\qI + \lambda_2\qh_2\qh_2^\dag\right)^{-1}\qh_1\|},\\
     \qw_2& = \sqrt{p_2}\frac{ \left(\sigma^2\qI + \lambda_1\qh_1\qh_1^\dag\right)^{-1}\qh_2}{\| \left(\sigma^2\qI + \lambda_1\qh_1\qh_1^\dag\right)^{-1}\qh_2\|},\\
     \end{array}\right.
\eea
where the downlink power $p_1, p_2$ can be found using the
uplink-downlink duality \cite{boche_algo}.} Then the optimal
solution to
 (\ref{eqn:prob:rate:max:AF:v1}) can be derived by performing 1-D optimization of $\rho$.
To efficiently find the optimal $\rho$,  we characterize  its feasible range in Appendix \ref{app:feas_pow_opt}.

 \subsection{Closed-form ZF Solutions}
  In order to gain more insight into the system parameters, we study the ZF solution
which allows a closed-form solution.
 According to the  ZF criterion, there should be no interference between the primary and the
secondary transmission, which requires that  $\qh_s^\dag\qw_p=
\qh_{sp}^\dag \qw_s=0$. The ZF solution to $\qw_s$ has been given in (\ref{eqn:DPC:ZF:ws})
and similarly, the ZF solution to $\qw_p$ can be derived
as
  \be
    \qw_{p,ZF} = \sqrt{q_p}\frac{ \left(\qI - \frac{\qh_{s}\qh_{s}^\dag}{\|\qh_{s}\|^2}
\right) \qh_{sp}}{\| \left(\qI - \frac{\qh_{s}\qh_{s}^\dag}{\|\qh_{s}\|^2}\right)
    \qh_{sp}\|},
 \ee
 with the resulting channel gain to the PU being
$|\qw_{p,ZF}^\dag\qh_{sp}|^2 =q_p\|\qh_{sp}\|^2(1-\delta^2)$.

Therefore problem (\ref{eqn:prob:rate:max:AF:v1}) reduces to
  \bea\label{eqn:prob:rate:max:AF:ZF:v0}
      \max_{q_p, q_s,\rho} && q_s \\
    \mbox{s.t.}
   &&  \frac{  q_p\|\qh_{sp}\|^2(1-\delta^2)}    {\tilde N_0}
    \ge \gamma_p^{''},  \notag\label{eqn:c1}\\
&&  q_s +  q_p \le   2P_{s0}+ \eta(1-\rho)(2 P_p\|\qg\|^2+ N_0),\notag\\
&& 0\le \rho\le 1, q_p\ge 0, q_s\ge 0. \notag
 \eea

The optimal $\rho_{zf}^*$ is given by
 \bea\label{eqn:rho:zf:ps}
    \rho_{zf}^* =\left[\frac{\frac{\sqrt{\gamma_{p}^{'}  N_C \tilde N_0}}{\sqrt{\|\qh_{sp}\|^2(1-\delta^2)}}\sqrt{   \gamma_{p}^{'}  + \frac{  (\|\qg\|^2-\gamma_{p}^{'} N_0)      }{\eta (2P_p\|\qg\|^2+ N_0)}} + \gamma_p' N_C}{\|\qg\|^2-\gamma_{0}^{'}
    N_C}\right]_0^1.
\eea
Derivation of the solution and its feasible range are given in Appendix \ref{app:opt_zf}.

We can draw some insights from (\ref{eqn:rho:zf:ps}) on {\bl{
$\rho_{zf}^*$}}:
 \begin{itemize}
    \item It increases with $\gamma_p'=\frac{2^{2 r_p}-1}{2P_p}-\frac{|h_p|^2}{\tilde N_0}$, or  decreases with $P_p$ and $|h_p|^2$,
    which means when the primary channel is good or power is abundant, PU rate is easy to satisfy, and thus
    the PT can transfer more energy to the ST.
    \item It decreases with $\eta$ and $\|\qh_{sp}\|^2$, which means if the efficiency of energy transfer is low or the ST-PU channel is weak, more
    received signal is used for information decoding. This is different from the ideal cooperation case where {\bl the} primary signal is non-causally known at the ST.
 \end{itemize}

  To illustrate the solutions, the achievable SU rates against the power splitting parameter $\rho$ for the optimal solution and ZF solution are compared in Fig. \ref{fig:cu:rate:rho} for a specific channel realization
 $N=3, |h_p|^2=0.0127, \qg = [   0.8113 - 1.5579i ~   0.4228 - 0.4039i~  -0.9060 + 0.1513i]^T, \qh_s = [   0.6664 + 0.2165i~   0.0663 - 0.8290i~  -0.7936 -
 0.6795i]^T$ and $\qh_{sp} = [  -0.4623 - 0.6364i   -0.8693 - 0.2020i  -0.1916 - 0.3270i]^T$.
 The primary power is set to $P_p=10$ dB and the ST's own power is $P_{s0}=0$ dB. The PU's target rate is  2.6 bps/Hz. All noise variance is normalized to one. It is clearly seen that the feasible
 range of the optimal solution includes that of the ZF solution as a subset. The optimal  SU rate is higher than  double of the ZF SU rate.

 \begin{figure}[]
        \centering
                 \includegraphics[width=3.4in]{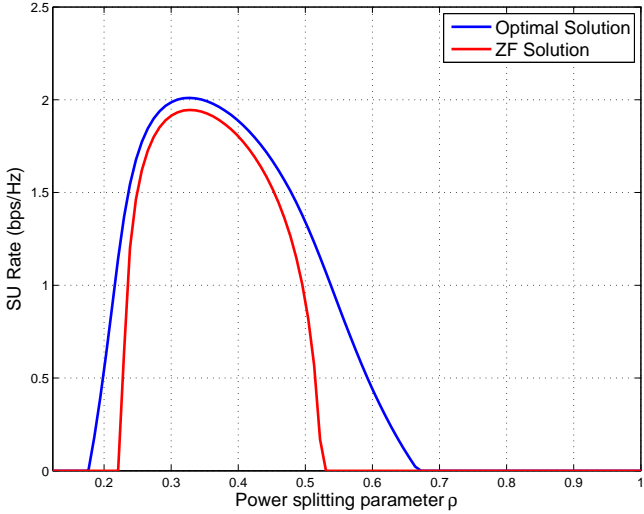}
        \caption{SU rate vs $\rho$ for power splitting. }
        \label{fig:cu:rate:rho}
\end{figure}

\section{Time Splitting Cooperation -- System Model and Optimization}

\subsection{System Model and Problem Setting}
 \begin{figure}[]
        \centering
                 \includegraphics[width=0.5\textwidth]{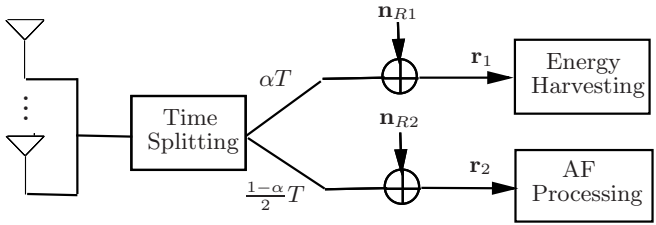}
        \caption{The time splitting EH technique  at the ST. }
        \label{fig:TS}
\end{figure}
 In this section, we study the optimization of a three-phase time-splitting cooperation protocol
where the time-splitting EH is illustrated in Fig. \ref{fig:TS}.
 The PT first uses a dedicated time slot with a duration of $\alpha ~(0\le \alpha\le 1)$ to transfer energy to the ST.
In the remaining two equal-time phases with duration
  of  $\frac{1-\alpha}{2}$, the PT transmits data to the ST then the ST
forwards the primary signal to the PU
 and serves its own SU. {  The PT can adjust its transmit power in the two phases as long as it does not exceed the peak power constraint $P_{max}$.
 The signal model is described below.}

{  In Phase I, the PT sends  signal $s_{p1}$ with average power
$P_{p1}$ to  both the ST for energy harvesting and the PU for
information decoding.} The received signal at the ST is
 \be
 \qr_1=  \qg  s_{p1} + \qn_{R1},\ee
 where    $\qn_{R1}  \in \mathcal{CN}(\qzero,\tilde N_0 \qI)$ is the   AWGN received at the ST.
The amount of the harvested energy is
 \be
    E_{EH} =  \alpha \eta (P_{p1} \|\qg\|^2  + N_0),
 \ee
 where $\eta$ is the efficiency of EH.

 {
 The PU receives
\be y_{p1}=    h_p s_{p1} + n_{p1},\ee where $ n_{p1}\in
\mathcal{CN}(0,\tilde
 N_0)$
 and achieves a rate of
    \be
        R_{p1} = \alpha  \log_2 \left(1+ \frac{P_{p1} |h_p|^2}{\tilde
        N_0}\right).
    \ee
}

 In Phase II, {  the PT sends signals $s_{p2}$ (${\tt E}[|s_{p2}|^2]=1$)  with average power {$P_{p2}$} to   the
 ST,}
  then the received signals at the PU and the ST are, respectively,
 \be y_{p2}=  \sqrt{P_{p2}} h_p s_{p2} + n_{p2},~~ \mbox{and}~~\qr_2=  \sqrt{P_{p2}}\qg  s_{p2} + \qn_{R2},\ee
 where  $ n_{p2}\in \mathcal{CN}(0,\tilde N_0)$ and  $\qn_{R2} \in \mathcal{CN}(\qzero, \tilde N_0\qI)$ are the combined noise received at the PU and the ST,
 respectively.
 The ST adopts the same strategy as the power splitting protocol to process the received primary signal, i.e., it  applies first an MRC receiver $\qg$
 then forwards it using a new beamforming vector $\qw_p$.

In Phase III, the ST superimposes the processed primary signal with its own data $s_s$   using the cognitive beamforming vector $\qw_s$, then
transmits it to both the PU and the SU. In this phase, the PT remains idle. The ST's transmit signal is written as
  \be
      \qt = \qw_s s_s+  \sqrt{P_{p2}}\qw_p\|\qg\|^2   s_{p2} + \qw_p\qg^\dag \qn_{R2},
 \ee
 with average power
  \be
    p_R = {\tt E} [\|\qt\|^2] =     \|\qw_s\|^2 +   \| \qw_p\|^2 \left( P_{p2} \|\qg\|^4
+  \|\qg \|^2 \tilde N_0\right).
     \ee
 The received signal at the SU is
 \bea
   y_s &=&\qh_{s}^\dag\qt + n_s\
    =    \qh_{s}^\dag\qw_s s_s \notag\\ &&+   \sqrt{P_{p2}}\qh_{s}^\dag \qw_p\|\qg\|^2  s_{p2}
+ \ \qh_{s}^\dag\qw_p\qg^\dag\qn_{R2}   +
    n_{s},
\eea where $n_s \in \mathcal{CN}(\qzero,\tilde N_0)$ is the combined  noise received at the SU.
 The received  SINR  at SU   is then expressed as
 \bea
    \Gamma_s    &=& \frac{|\qh_{s}^\dag\qw_s|^2}{    P_{p2}   | \qh_{s}^\dag\qw_p|^2\|\qg\|^4  +
     |\qh_{s}^\dag\qw_p|^2\|\qg\|^2\tilde N_0 +  \tilde N_0 },
\eea and the achievable rate is $R_s =
\frac{1-\alpha}{2}\log_2(1+\Gamma_s)$. The received signal at the PU
is
 \bea
     y_{p3} &=&\qh_{sp}^\dag\qt+n_{p3} =  \qh_{sp}^\dag  \qw_s s_s+  \sqrt{P_{p2}}\qh_{sp}^\dag \qw_p\|\qg\|^2 s_{p2}  \notag\\
    &&+    \qh_{sp}^\dag \qw_p\qg^\dag\qn_{R2}    +   n_{p2}.
\eea

 Applying   MRC  to $y_{p2}$ and $y_{p3}$, the received SINR of the PU is the sum  of two channel uses,
 and {  considering $R_{p1}$ in the first phase}, the total  PU rate is
 \bea
    R_p &= &\alpha  \log_2 \left(1+ \frac{P_{p1} |h_p|^2}{\tilde
        N_0}\right) + \frac{ 1-\alpha  }{2}\log_2\Big(1+\frac{P_{p2} |{h}_p|^2}{\tilde N_0}
    \notag \\&&+  \frac{  P_{p2}  | \qh_{sp}^\dag\qw_p|^2\|\qg\|^4}  { |\qh_{sp}^\dag\qw_s|^2+    |\qh_{sp}^\dag \qw_p|^2\|\qg\|^2\tilde N_0   + \tilde N_0}\Big).
 \eea

 The problem of maximizing the SU rate with PU rate and ST power constraint is formulated
 as (\ref{eqn:prob:TS}) at the top of next page,
 \begin{figure*}
 \bea\label{eqn:prob:TS}
    \max_{\alpha, P_{p1},  \qw_s,\qw_p}&& \frac{1-\alpha}{2}\log_2\left(1+\frac{|\qh_{s}^\dag\qw_s|^2}{    P_{p2}   | \qh_{s}^\dag\qw_p|^2\|\qg\|^4  +
     |\qh_{s}^\dag\qw_p|^2\|\qg\|^2\tilde N_0  +  \tilde N_0 }\right)\\
     \mbox{s.t.} && \alpha  \log_2 \left(1+ \frac{P_{p1} |h_p|^2}{\tilde
        N_0}\right) + \frac{ 1-\alpha }{2}\log_2\left(1+\frac{P_{p2} |{h}_p|^2}{\tilde N_0}
    +  \frac{  P_{p2}  | \qh_{sp}^\dag\qw_p|^2\|\qg\|^4}  { |\qh_{sp}^\dag\qw_s|^2+    |\qh_{sp}^\dag \qw_p|^2\|\qg\|^2\tilde N_0    + \tilde N_0}\right)\ge r_p,\notag\\
&&  P_{p1}\le P_{max}, P_{p2} = \max\left(P_{max},2\frac{ P_p - \alpha P_{p1}}{1-\alpha}\right),\notag\\
       && \|\qw_s\|^2 +   \| \qw_p\|^2 \left( P_{p2} \|\qg\|^4  +  \|\qg \|^2 \tilde N_0 \right)    \le 2 \frac{  \alpha \eta  (P_{p1} \|\qg\|^2  + N_0) + P_{s0}}{(1-\alpha)},\notag\\ &&0\le \alpha\le 1. \notag
     \eea
     \end{figure*}
    {  where we have imposed the peak power constraint $P_{max}$ on the
    PT's transmit power in both Phase I and Phase II, to prevent
    extremely high transmit power.}

{  Following a similar procedure to obtain
(\ref{eqn:prob:rate:max:AF:v1})}, problem (\ref{eqn:prob:TS}) can be
written more compactly as
 {\small \bea\label{eqn:prob:TS2}
    \max_{\alpha, P_{p_1},   \qw_s,\qw_p}&& \frac{1-\alpha}{2}\log_2\left(1+\frac{|\qh_{s}^\dag\qw_s|^2}{ |\qh_{s}^\dag\qw_p|^2  +  \tilde N_0 }\right)\\
     \mbox{s.t.} && \frac{  | \qh_{sp}^\dag\qw_p|^2}  { |\qh_{sp}^\dag\qw_s|^2 + \tilde N_0}\ge \frac{\left( P_{p2} \|\qg\|^2  +   \tilde N_0 \right)\gamma_p' }
     { \|\qg\|^2  - \gamma_p'  \tilde N_0 },\notag\\
     && P_{p1}\le P_{max}, P_{p2} = \max\left(P_{max},2\frac{ P_p - \alpha P_{p1}}{1-\alpha}\right),\notag\\
       && \|\qw_s\|^2 +   \| \qw_p\|^2     \le 2\frac{ \alpha \eta (P_{p1} \|\qg\|^2  + N_0)+ P_{s0}}{(1-\alpha)} ,\notag\\
       && 0\le \alpha\le 1,\notag
     \eea}
 {\bl   where {  $\gamma_p'\triangleq \frac{2^{\frac{2\left(r_p - \alpha \log_2 \left(1+ \frac{P_{p1} |h_p|^2}{\tilde
        N_0}\right) \right)}{(1-\alpha)}}-1}{P_{p2}}-\frac{|h_p|^2}{\tilde N_0}$}.
Given $\alpha$ and $P_{p1}$, the optimal $\qw_s, \qw_p$ can be found
using Lemma \ref{lemma:gam2}. Therefore (\ref{eqn:prob:TS2}) can be
  solved by
 performing  2-D search over $(\alpha, P_{p1})$.}

\subsection{ZF Solution}
 Similar to the case of power splitting, we study the ZF solution for time splitting which require that  $\qh_s^\dag\qw_p= \qh_p^\dag \qw_s=0$.
 The simplified problem (\ref{eqn:prob:TS2})  becomes
  {\small\bea\label{eqn:ts:ZF}
    \max_{\alpha, P_{p1},  q_s,q_p}&& \frac{1-\alpha}{2}\log_2\left(1+\frac{q_s\|\qh_{s}\|^2(1-\delta^2)}{    \tilde N_0 }\right)\\
     \mbox{s.t.} &&  \frac{  q_p\|\qh_{sp}\|^2(1-\delta^2) }  {  \tilde
        N_0}  \ge\frac{\left( P_{p2} \|\qg\|^2  +   \tilde N_0\right)\gamma_p' }{ \|\qg\|^2  - \gamma_p'  \tilde N_0},\notag\\
         && P_{p1}\le P_{max}, P_{p2} = \max\left(P_{max},2\frac{ P_p - \alpha P_{p1}}{1-\alpha}\right),\notag\\
       && q_s +   q_p     \le 2\frac{ \alpha \eta (P_{p1} \|\qg\|^2  + N_0)+ P_{s0}}{(1-\alpha)},\notag\\
       && 0\le \alpha\le 1, q_s\ge 0, q_p\ge 0.\notag
     \eea}

 {
  Given $(\alpha, P_{p1})$,  the solution to $(q_s, q_p)$ is easily  derived. However, this leads to a complicated objective function about $(\alpha, P_{p1})$ which does not admit a
closed-form solution, therefore the optimal solution   can be found
by performing 2-D search over $(\alpha, P_{p1})$.}

\section{Simulation Results}
 {\bl In this section} the performance results of the proposed  primary-secondary
cooperation schemes are presented through computer simulation.
 We assume that the ST has $N=4$ transmit antennas. {  We consider a scenario where the distances from the ST to all the other terminals  are 1m,
 while the distance from the PT to the
 PU is 2m, therefore assistance from the ST is usually preferred by the
 PT. The potential application scenarios include wireless sensor networks or the indoor environment where WiFi and ZigBee coexist
 both operating at 2.4 GHz which leads to significant interference \cite{wifi}.   WiFi is the primary system and ZigBee is the secondary system. Zigbee wants to share the
 spectrum occupied by WiFi but it has very limited energy supply or   even no battery.}
 The channel between a transmit-receive antenna pair is modeled as {\bl{$h=(\Delta)^{-\frac{l}{2}}e^{j\omega}$, where $\Delta$}} is the distance, $l$ is the path loss exponent, chosen as $3.5$,
 and $\omega$ is uniformly distributed over $[0,2\pi)$. The variance of   noise components are normalized to unity, i.e., $N_0=N_C=1$.
 The primary {  energy} is set to $P_p=20$ dB,  {  the  peak power is $P_{max} = 30$ dB, and
  the PU rate requirement is $r_p=3$ bps/Hz,  unless otherwise specified.}
  Outage  occurs when the required PU rate is  not supported.
We will evaluate the performance of the proposed schemes including the ideal cooperation
 (labeled as 'Ideal Cooperation'),
 the power splitting scheme
 (labeled as 'Power Splitting EH') and the time splitting scheme
  with both fixed equal power  (labeled as `Time Splitting EH' ) and adaptive power
allocation (labeled as `Time Splitting EH with Power Allocation') during the energy
and information transfer.
  The case of information cooperation only \cite{Zheng_CCRN_13}
   between the primary and secondary systems without energy cooperation is used as the
 benchmark (labeled as 'No Energy Cooperation'). {\bl Unless otherwise
 specified, the results are averaged over 1000 channel
 realizations.}

In Fig. \ref{fig:rate:region}, we plot the rate regions for
different schemes  for   a specific randomly chosen channel
realization $|h_p|^2=0.0002, \qg = [
  -0.9472 - 0.6334i~
  -0.9090 - 1.2266i~
  -1.1855 + 0.3370i~
   0.5345 - 0.1796i]^T, \qh_s = [    -0.9215 - 0.4314i~
   0.2052 - 0.2503i~
   0.3109 - 0.3055i~
   0.3560 + 0.1163i]^T$ and $\qh_{sp}
= [     0.3610 - 0.1248i~
   1.1616 + 0.8211i~
  -0.4350 - 0.2818i~
  -0.4445 + 0.6564i]^T$.  ST's own energy is $P_{s0}=10$ dB. {  The efficiency of energy transfer is assumed to be
  $\eta=0.5$.}
It is seen that the achievable rate regions are greatly enlarged due
to the energy cooperation. {  It is observed that the power
splitting scheme for energy cooperation outperforms the time
splitting scheme. }  Due to the non-causal information transfer, the
ideal information and energy cooperation provides an outer bound for
both practical cooperation schemes.

Next we investigate the impact of the ST's self energy on the
achievable {  average} SU rate in Fig. \ref{fig:cu:rate:cu:power} {
when the efficiency of energy transfer, $\eta$, takes values 0.1,
0.5 and 1. Substantial rate gain is achieved using the proposed
schemes compared with the case without energy cooperation,
especially in low to medium ST energy region even  when the
efficiency is $\eta=0.1$. While in the high energy region and
$\eta=0.1$, the performance of the information only cooperation
scheme is close to that of  the power splitting and time splitting
scheme. This is because the ST has sufficient energy and the
efficiency of energy transfer is low, there is no need to harvest
energy from the primary transmission. It is also observed that the
power splitting scheme achieves higher SU rate than the time
splitting scheme when the efficiency $\eta=0.5$ and 1 while when
$\eta=0.1$, the time splitting scheme outperforms the power
splitting scheme in the low to medium energy region.}

We then compare the PU rate outage performance of different schemes
when the ST energy varies  from 0 to 20 dB in Fig.
\ref{fig:pu:outage:cu:power}.  We assume   the SU rate requirement
is $r_s=3$ bps/Hz. It is first noted that without any
primary-secondary cooperation, the PU experiences rate outage with a
high probability of over  80\% due to the weak primary channel. If
only information cooperation but no energy transfer is allowed, the
outage probability can be reduced only in the high  energy region.
 {  When
  $\eta\ge 0.5$, the  power splitting and time splitting
schemes achieve outage probabilities of below $20\%$ and $35\%$,
respectively, which is a substantial improvement.}

 \begin{figure}[!ht]
        \centering
                 \includegraphics[width=3.5in]{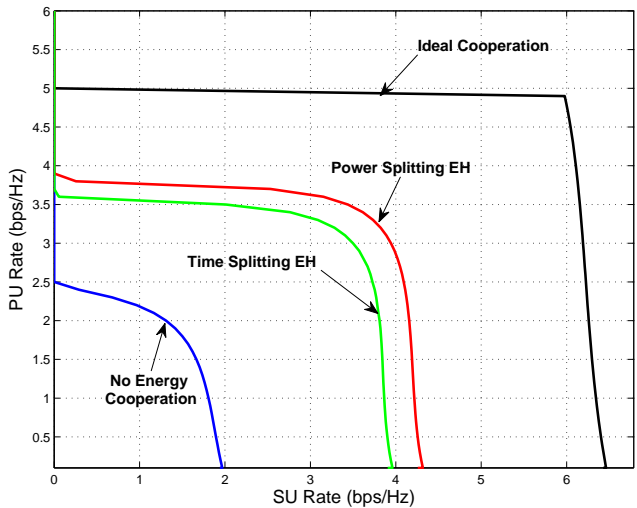}
        \caption{PU-SU rate region, $P_p=20$ dB, $P_{s0}=10$ dB. }
        \label{fig:rate:region}
\end{figure}

 \begin{figure}[!ht]
        \centering
                 \includegraphics[width=3.5in]{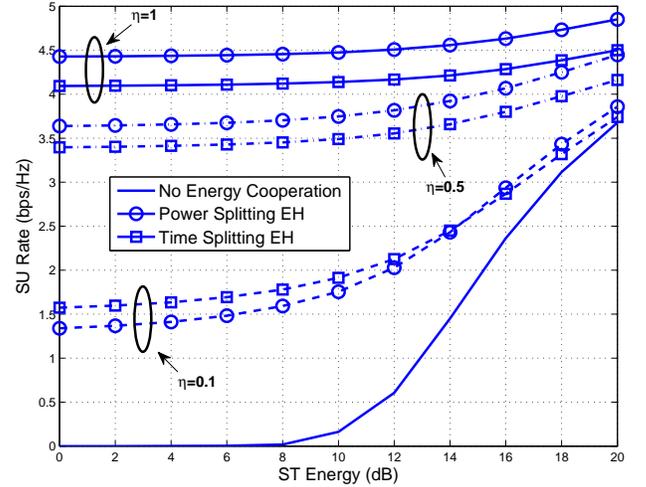}
        \caption{SU rate vs SU power, $r_p=3$ bps/Hz. }
        \label{fig:cu:rate:cu:power}
\end{figure}

 \begin{figure}[!ht]
        \centering
                 \includegraphics[width=3.5in]{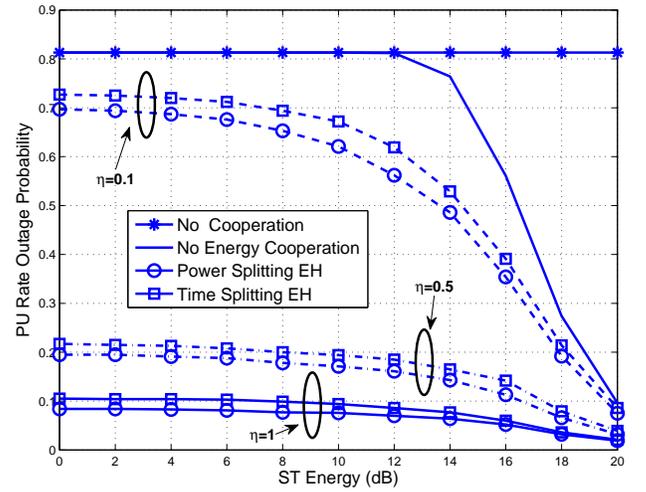}
        \caption{PU rate outage vs SU power, $r_p=3$ bps/Hz   and $r_s=4$ bps/Hz. }
        \label{fig:pu:outage:cu:power}
\end{figure}

\section{Conclusions}
 This paper has investigated energy cooperation between the primary and secondary system in cognitive radio networks,  in addition to existing information
 cooperation. The rationale behind is that the primary system provides spectrum as well energy to the secondary system, and in return the secondary system
 is more willing and able to  assist the primary transmission. This creates more incentives for the primary and secondary systems to cooperate and overall the
 spectrum is better utilized. We have studied three protocols that enables both energy and information cooperation.
 The first one is the ideal cooperation assuming non-causal primary information available at the secondary transmitter; the other two
 protocols employ   practical power splitting and time splitting for energy and information transfer.
 For each scheme, the optimal as well as low-complexity solutions are derived, based on which some insights are drawn on system parameters.
 {  It is found that  the power splitting scheme usually can support a larger rate region  than the time splitting scheme when the efficiency of energy transfer
 is sufficiently high.} Substantial
performance gain has been shown using the proposed
 additional energy cooperation than the existing information cooperation only CR scheme, therefore  energy and information cooperation could be a promising
 solution for the future CR networks.

\section*{Acknowledgement}
The authors thank the anonymous reviewers and the associated editor
for their useful comments which have substantially improve the
quality  of the manuscript. The first and the third authors would
like to thank  David Halls and  Aissa Ikhlef from Toshiba Research
Europe Limited, for their suggestions to improve the clarity of this
manuscript.

\appendices

\section{Proof of Proposition \ref{prop:feas}}\label{app:feas}
The maximum PU rate is achieved when SU rate is zero and $\qw_s=\qzero$. In this case, $q_p=P_{s0} + \beta \eta P_p$ and the optimal $\beta$ that  maximizes
 the  PU rate is given by
 \be
    \beta^* = \arg_{\beta} \max_{0\le \beta\le 1} \left(\sqrt{(1-\beta) P_p}|h_p| + \sqrt{P_{s0} + \beta \eta P_p}\|\qh_{sp}\| \right).
 \ee
 Setting the derivative to zero leads to the unique critical point
 \be
    \bar\beta=\frac{P_p \eta^2 \|\qh_{sp}\|^2 - P_{s0} |h_p|^2}{ P_p
\eta^2 \|\qh_{sp}\|^2 + \eta P_p |h_p|^2}.
 \ee
  Then the optimal $\beta$ is $\beta^*=[\bar\beta]_0^1.$
 Accordingly, the maximum PU rate in (\ref{eqn:Rpmax:DPC}) can be
 derived.
%
%
%
{{
\section{Recasting \eqref{eqn:prob:DPC} to the convex program (\ref{eqt:cov_opt})}\label{app:convex}
We begin by realizing that the optimization problem relates to
$\w_s$ only in quadratic forms and we can multiply a complex phase
to the optimization objective so that the complex phase is offset.
Hence, we replace the optimization objective by
$\re\left(\h_s^{\dagger} \w_s\right)$.
\begin{equation}
 \begin{aligned}
  \max_{\w_s, \w_p, \beta} \hspace{1cm} & \re\left(\h_s^{\dagger} \w_s\right)\\
  \st \hspace{1cm} & |\sqrt{(1-\beta) P_p} |h_p| + \h_{sp}^{\dagger}
  \w_p|^2\notag \\
&\geq \gamma_p \left( |\h_{sp}^{\dagger} \w_s|^2 + \tilde{N}_0 \right),\\
 & \| \w_s\|^2 + \|\w_p\|^2 \leq P_{s0} + \beta \eta P_p,
 \end{aligned}
\end{equation} where $\gamma_p= 2^{R_p}-1$ is the target SINR.
The optimal direction of $\w_p=\sqrt{q_p} \frac{\h_{sp}}{\|\h_{sp}\|} e^{j \arg(h_p)}$ can be determined
and what is left for optimizing is the power of $\w_p$, $q_p$. Thus we have
\begin{equation}
 \begin{aligned}
  \max_{\w_s, q_p, \beta} \hspace{1cm} & \re\left(\h_s^{\dagger} \w_s\right)\\
  \st \hspace{1cm} & (\sqrt{(1-\beta) P_p} |h_p| + \sqrt{q_p} \|\h_{sp}\| )^2\notag \\
&\geq \gamma_p \left( |\h_{sp}^{\dagger} \w_s|^2 + \tilde{N}_0 \right),\\
 & \| \w_s\|^2 + q_p \leq P_{s0} + \beta \eta P_p.
 \end{aligned}
\end{equation}
Now, we let
$\bar{\beta}= \sqrt{1-\beta}$ and
 $\bar{q}_p = \sqrt{q_p}$ and the ranges of the parameters remain the same: $0 \leq \bar{\beta} \leq 1$ and $\bar{q_p} \geq 0$.
The optimization problem becomes
\begin{equation}
 \begin{aligned}
  \max_{\w_s, \bar{q}_p, \bar{\beta}} \hspace{1cm} & \re\left(\h_s^{\dagger} \w_s\right)\\
  \st \hspace{1cm} & ( \bar{\beta} \sqrt{ P_p} |h_p| + \bar{q}_p \|\h_{sp}\| )^2
\geq \gamma_p \left( |\h_{sp}^{\dagger} \w_s|^2 + \tilde{N}_0 \right),\\
 & \| \w_s\|^2 + \bar{q}_p^2 \leq P_{s0} + (1-\bar{\beta}^2) \eta P_p,\\
& 0 \leq \bar{\beta} \leq 1, \; \bar{q_p} \geq 0.
 \end{aligned}
\end{equation}
Define a vector $\bv$:
\begin{equation}
 \bv=\left[ \bar{\beta}, \bar{q}_p\right]^{\tran}
\end{equation}
and the result follows. }

\section{Proof of Proposition \ref{prop:2d}}\label{app:feas2}

    Define $q_s\triangleq \|\qw_s\|^2$. First notice that the PU rate and ST power constraints in (\ref{eqn:prob:DPC}) should be satisfied with
    equalities. So we have
    \be\label{eqn:qs}
        q_s = (P_{s0} + \beta \eta P_p) - q_p.
    \ee
According to Proposition 1 in  \cite{Jorswieck-WSA-11}, the optimal $\qw_s$ can be parametrized as
  \be
  \qw_s=  \sqrt{q_s} \left(\sqrt{\lambda} \frac{\Pi_{\qh_{sp}} \qh_s}{\|\Pi_{\qh_{sp}} \qh_s\|} + \sqrt{1-\lambda} \frac{\Pi_{\qh_{sp}}^\bot \qh_s}{\|\Pi_{\qh_{sp}}^\bot \qh_s\|}\right), 0\le \lambda\le
  1.
  \ee
 We then have $|\qh_{sp}^\dag \qw_s |^2 =  \lambda q_s\|\qh_{sp}\|^2$. Substituting it into the PU rate constraint in (\ref{eqn:prob:DPC}), we can
 solve   $\lambda$ as:
\be\label{eqn:lambda}
    \lambda =  \frac{\frac{|\sqrt{(1-\beta) P_p}|h_p| + \sqrt{q_p}\|\qh_{sp}\| |^2}  {2^{r_p}-1 }-\tilde N_0}{((P_{s0} + \beta \eta P_p) -
    q_p)\|\qh_{sp}\|^2}.
\ee
Substituting (\ref{eqn:qs})--(\ref{eqn:lambda}) into the objective function in  (\ref{eqn:prob:DPC}) gives the formulation
 (\ref{eqn:prob:DPC:2D}).

%
\section{Feasibility {\bl range} in power-splitting cooperation}\label{app:feas_pow}
Define $a_1\triangleq 2 P_p   \|\qg\|^2+   N_0 ,~~b_1\triangleq
2P_{s0},~~c_1\triangleq \|\qh_{sp}\|^2$, then
(\ref{eqn:prob:rate:max:AF:v4}) is equivalent to
      \bea\label{eqn:prob:rate:max:AF:v5}
      \min_{\rho} &&  f_1(\rho)\triangleq \frac{\tilde N_0(a_1+\frac{N_C}{\rho})}{c_1[b_1+a_1\eta(1-\rho)]}+(N_0+\frac{N_C}{\rho})\\
    \mbox{s.t.}&& 0\le \rho\le 1.\notag
 \eea
The first-order derivative of $f_1(\rho)$ is given by \bea
   &&\frac{\partial  f_1(\rho)}{\partial \rho} \\
          &=&-\frac{\tilde N_0N_C}{c_1[b_1+a_1\eta(1-\rho)]\rho^2}+
  \frac{\tilde N_0(a_1\rho+N_C)a_1\eta}{c_1[b_1+a_1\eta(1-\rho)]^2\rho}-\frac{N_C}{\rho^2},\notag
\eea and setting it to zero leads to
 \be f_2(\rho)\triangleq A_1\rho^2+B_1\rho-C_1=0,\ee
 where for convenience, we have  defined $A_1\triangleq a_1^2\eta \frac{\tilde N_0}{N_C}-(a_1\eta)^2c_1, B_1\triangleq 2a_1\eta N_0+2(b_1+a_1\eta)a_1\eta c_1, C_1\triangleq
N_0(b_1+a_1\eta)+c_1(b_1+a_1\eta)^2$,
%

 Observe that $f_2(\rho)$ has the same sign as $\frac{\partial  f(\rho)}{\partial \rho}$ in order to find the optimal $\rho^*$. Next we discuss the roots of $f_2(\rho)$.
 Because
    \bea
    && B_1^2+4A_1C_1 = (2a_1\eta N_0+2(b_1+a_1\eta)a_1\eta c_1)^2 \notag \\
     &&+ 4 \left(a_1^2\eta \frac{\tilde N_0}{N_C}-(a_1\eta)^2c_1\right) (N_0(b_1+a_1\eta)\notag  +c_1(b_1+a_1\eta)^2)\notag\\
      &>&  (2a_1\eta N_0+2(b_1+a_1\eta)a_1\eta c_1)^2 \notag \\ && - 4 (a_1\eta)^2c_1 (N_0(b_1+a_1\eta)+c_1(b_1+a_1\eta)^2)\notag\\
      &=& 4 a_1^2\eta^2 \left(N_0 ^2  + c_1 N_0 (b_1+ a_1\eta)\right)>0,
     \eea
 there are always two distinct real roots
 $\rho_1=\frac{-B_1+ \sqrt{B_1^2+4AC}}{2A_1}$ and $\rho_2=\frac{-B_1-\sqrt{B_1^2+4AC}}{2A_1}$. Depending on the sign of $A_1$, there are three possible cases for the optimal $\rho$:
 \begin{itemize}
    \item[i)] $A_1>0$ or $\frac{\tilde N_0}{N_C}> \eta\|\qh_{sp}\|^2$. Because $C_1>0$, $\rho_1>0, \rho_2<0$, therefore $\rho^* = \min( \rho_1 ,1)$.
    \item[ii)] $A_1<0$ or $\frac{\tilde N_0}{N_C}< \eta\|\qh_{sp}\|^2$. In this case, $\rho_2>\rho_1>0$.  Due to the fact that  $A_1+C_1 >0$ or $\frac{-C_1}{A_1}=  \rho_1\rho_2>1$, we know that $\rho_2>1$ and it cannot be the optimal solution.
    Therefore, $\rho^* = \min( \rho_1,1)$.
    \item[iii)] $A_1=0$. In this case, $\rho^*= \min\left(\frac{C_1}{B_1},1\right)=\min\left(\frac{P_{s0}}{2(2P_p   \|\qg\|^2+   N_0)\eta} + \frac{1}{2},1\right).$
 \end{itemize}
 After finding $\rho^*$, the maximum achievable PU rate can be calculated and compared with the PU rate requirement to check the feasibility.

\section{Feasibility Range of $\rho$ in power splitting cooperation}\label{app:feas_pow_opt}

  A given $\rho$ can result in a feasible solution only if the ST can satisfy the PU's rate requirement even without serving the SU. In this extreme
  case, $\qw_s=0$ and $\qw_p= \frac{\sqrt{2P_{s0}+P_{EH}}\qh_{sp}}{\|\qh_{sp}\|}$. The PU's rate constraint then amounts to
  \bea\label{eqn:fea}
   (2P_{s0}+ \eta(1-\rho)(2P_p\|\qg\|^2+ N_0))\|\qh_{sp}\|^2  \notag \\
   \ge  \frac{(2P_p \rho \|\qg\|^2 + \rho N_0  + N_C)\gamma_{p}^{'} \tilde N_0}{ \rho\|\qg\|^2-\gamma_{p}^{'}(\rho N_0 +
   N_C)}.
   \eea

 For convenience, define $a\triangleq \eta(2P_p\|\qg\|^2+ N_0), b\triangleq   \|\qg\|^2-\gamma_{p}^{'} N_0, c \triangleq \frac{\gamma_{p}^{'}
\tilde N_0}{\|\qh_{sp}\|^2}$ and  $A   \triangleq ab, B \triangleq
-a \gamma_p' N_C -b(2P_{s0}+a)+ac, C  \triangleq
N_C(\gamma_p'(2P_{s0}+a) + c)$.
 Rearranging the above inequality (\ref{eqn:fea}) leads to
   \bea
  2P_{s0} + a - a\rho - c \frac{\frac{a\rho}{\eta} + N_C}{b\rho-\gamma_p' N_C}\ge 0,
 \eea
 and
  \be\label{eqn:fea:opt:PS}
    f_{PS}(\rho) \triangleq A\rho^2 + B\rho + C\le 0.
  \ee
 We then discuss possible cases below.

 \begin{itemize}
    \item $A=0$ or $\|\qg\|^2 =  N_0(\frac{2^{2 r_p}-1}{2P_p}-\frac{|h_p|^2}{\tilde N_0})$. Under this condition, there are two possibilities:
        \begin{itemize}
            \item $c\ge \gamma'_p \eta N_C$ or  $  \tilde N_0 \ge  \eta N_C \|\qh_{sp}\|^2$. In this case, $0\le \rho \le 1$.
            \item $c <\gamma'_p \eta N_C$ or  $  \tilde N_0 < \eta N_C \|\qh_{sp}\|^2$, then
            $0\le \rho<\left[\frac{2 P_{s0} + a  + \tilde N_0 \|\qh_{sp}\|^2}{a - \frac{\tilde N_0 \|\qh_{sp}\|^2a}{\eta N_C}}\right]_0^1.$
        \end{itemize}
   When $A\ne 0$, there are two possible cases:
   \item
    If the discriminant is nonnegative,  suppose the two real roots of the quadratic equation $f_{PS}=0$ are given by $\bar\rho_{\min}$ and
    $\bar\rho_{\max}$, and $\bar\rho_{\min}\le \bar\rho_{\max}$. We have the following discussion.
       \begin{itemize}
    \item $A>0$. The feasible $\rho$ should satisfy $[\bar\rho_{\min}]_0^1 \le \rho\le [\bar\rho_{\max}]_0^1$.
    \item $A<0$. Because $C>0$, we know that $\bar\rho_{\min}<0$ and $\bar\rho_{\max}>0$, thus we have $[\bar\rho_{\max}]_0^1 \le \rho \le 1$.
        \end{itemize}
    \item
    If the discriminant is negative, there is no feasible $\rho$.
 \end{itemize}

\section{Optimality of $\rho_{zf}^*$ in power splitting cooperation}\label{app:opt_zf}
Because the first two constraints in (\ref{eqn:prob:rate:max:AF:ZF:v0})
should hold with equality, we can derive
 \bea
    q_s(\rho)& =& 2P_{s0} + a - a\rho - c_{zf} \frac{a\rho + N_C}{b\rho-\gamma_p' N_C}\\
            &=& 2 P_{s0} + a - a\rho -  \frac{ac}{b} - \frac{(\gamma_p'\frac{a}{b}+1) N_C c_{zf}}{b\rho-\gamma_p' N_C},\notag \eea
            where $c_{zf} =  \frac{\gamma_{p}^{'} \tilde N_0}{\|\qh_{sp}\|^2(1-\delta^2)}$.
Maximizing $q_s$ is equivalent to \be
    \min_{0<\rho<1} a\rho + \frac{(\gamma_p'\frac{a}{b}+1) N_C c_{zf}}{b\rho-\gamma_p' N_C}
\ee whose minimum is achieved by \eqref{eqn:rho:zf:ps}.
%
 To compare with the optimal solution, we also study the feasible range of $\rho$ for the ZF scheme.
 A given $\rho$ is feasible requires that $q_s(\rho)\ge 0$ or
   \be\label{eqn:fea:opt:PS}
    f_{PS-ZF}(\rho) \triangleq A\rho^2 + B\rho +  C_{zf} \le 0,
  \ee
  where $C_{zf}  \triangleq N_C(\gamma_p'(2 P_{s0}+a) +c_{zf})$.
  Similar results as the optimal case (\ref{eqn:fea:opt:PS}) can be obtained, except replacing $C$ with $C_{zf}$. Since  $C_{zf}>C$ due to
  the fact that  $c_{zf}>c$, the optimal scheme has a larger feasibility region for $\rho$ than the ZF solution.

 \end{document}